\documentclass[reqno]{amsart}

\usepackage{booktabs}
\usepackage{IEEEtrantools}
\usepackage{amssymb,latexsym,amsfonts,amsmath}
\usepackage{mathrsfs}
\usepackage{graphicx}
\usepackage{dsfont}
\usepackage{tikz}
\usetikzlibrary{calc,shapes,arrows}

\definecolor{mediumblue}{rgb}{0.0, 0.0, 0.8}
\definecolor{mediumcandyapplered}{rgb}{0.89, 0.02, 0.17}
\definecolor{nazar}{rgb}{0.7, 0.5, 0.9}
\makeatletter
\let\NAT@parse\undefined
\makeatother
\usepackage[colorlinks=true, citecolor=black, linkcolor=mediumblue, urlcolor = nazar, final]{hyperref}

\definecolor{lightblue}{rgb}{0.30,0.75,0.93}
\newcommand{\greensquare}{\tikz\fill[green!20!white] (0,0) rectangle (2mm,2mm);}
\newcommand{\redsquare}{\tikz\fill[red!20!white] (0,0) rectangle (2mm,2mm);}
\newcommand{\reddsquare}{\tikz\fill[red!70!white] (0,0) rectangle (2mm,2mm);}
\newcommand{\bluesquare}{\tikz\fill[lightblue!60!white] (0,0) rectangle (2mm,2mm);}

\DeclareRobustCommand\sampleline[1]{%
	\tikz\draw[#1] (0,0) (0,\the\dimexpr\fontdimen22\textfont2\relax)
	-- (2em,\the\dimexpr\fontdimen22\textfont2\relax);%
}

\usepackage{algorithm}
\usepackage{algorithmic}

\usepackage{xspace}

\newtheorem{theorem}{Theorem}[section]
\newtheorem{lemma}[theorem]{Lemma}

\newtheorem{problem}[theorem]{Problem}

\newtheorem{definition}[theorem]{Definition}

\newtheorem{remark}[theorem]{Remark}

\numberwithin{equation}{section}

\usepackage{fancyhdr}

\newenvironment{nouppercase}{%
	\renewcommand{\uppercasenonmath}[1]{}}{}

\usepackage{amsmath}
\usepackage[many]{tcolorbox}
\usetikzlibrary{calc}
\tcbuselibrary{skins}

\newtcolorbox{resp}[1][]{%
	enhanced jigsaw,%
	colback=gray!5!white,%
	colframe=gray!80!black,%
	size=small,%
	boxrule=1pt,%
	halign title=flush center,%
	coltitle=black,%
	breakable,%
	drop shadow=black!50!white,%
	attach boxed title to top left={xshift=1cm,yshift=-\tcboxedtitleheight/2,yshifttext=-\tcboxedtitleheight/2},%
	minipage boxed title=3cm,%
	boxed title style={%
		colback=white,%
		size=fbox,%
		boxrule=1pt,%
		boxsep=2pt,%
		underlay={%
			\coordinate (dotA) at ($(interior.west) + (-0.5pt,0)$);
			\coordinate (dotB) at ($(interior.east) + (0.5pt,0)$);
			\begin{scope}[gray!80!black]
				\fill (dotA) circle (2pt);
				\fill (dotB) circle (2pt);
			\end{scope}
		}%
	},%
	#1%
}

\begin{document}

\begin{abstract}
This work is concerned with developing a data-driven approach for learning \textit{control barrier certificates} (CBCs) and associated \textit{safety controllers} for discrete-time nonlinear polynomial systems with \textit{unknown mathematical models}, guaranteeing system safety over an \textit{infinite time horizon}. The proposed approach leverages measured data acquired through an input-output observation, referred to as \textit{a single trajectory}, collected over a specified time horizon. By fulfilling a certain rank condition, which ensures the unknown system is \textit{persistently excited} by the collected data, we design a CBC and its corresponding safety controller directly from the finite-length observed data, \textit{without explicitly identifying} the unknown dynamical system. This is achieved through proposing a data-based sum-of-squares optimization (SOS) program to systematically design \textit{CBCs} and their safety controllers. We validate our data-driven approach over two physical case studies including a jet engine and a Lorenz system, demonstrating the efficacy of our proposed method.
\end{abstract}

\title{{\LARGE From a Single Trajectory to Safety Controller Synthesis of\vspace{0.2cm}\\ Discrete-Time Nonlinear Polynomial Systems\vspace{0.4cm}}}

\author{{\bf {\large Behrad Samari}}}
\author{{\bf {\large Omid Akbarzadeh}}}
\author{{\bf {\large Mahdieh Zaker}}}
\author{{{\bf {\large Abolfazl Lavaei}}}\vspace{0.4cm}\\
	{\normalfont School of Computing, Newcastle University, United Kingdom}\\
\texttt{\{b.samari2, o.akbarzadeh2, m.zaker2, abolfazl.lavaei\}@newcastle.ac.uk}}

\pagestyle{fancy}
\lhead{}
\rhead{}
  \fancyhead[OL]{Behrad Samari, Omid Akbarzadeh, Mahdieh Zaker, and Abolfazl Lavaei}

  \fancyhead[EL]{From a Single Trajectory to Safety Controller Synthesis of Discrete-Time Nonlinear Polynomial Systems}
  \rhead{\thepage}
 \cfoot{}
 
\begin{nouppercase}
	\maketitle
\end{nouppercase}

\section{Introduction}\label{sec: intro}
Nonlinear dynamical systems are fundamental to the scientific understanding of almost all practical real-world systems, particularly those with \emph{safety-critical applications}. However, analyzing and synthesizing controllers for such complex systems has always been a challenging task—especially when their precise mathematical models are \textit{unknown}, a situation that reflects the complexities of the modern world. As a result, the control community has increasingly leaned toward \textit{data-driven} approaches for control analysis and design, where the desired controller is learned and derived from collected data. Nonetheless, providing \textit{formal guarantees} in data-driven control for safety-critical systems that lack precise mathematical models remains a significant challenge.

In addressing the challenge of formally synthesizing controllers for unknown dynamical systems, two distinct yet feasible approaches have been proposed in the relevant literature: (i) integrating system identification with model-based techniques (known as the \textit{indirect} data-driven approach), and (ii) deriving the desired controllers directly from data (known as the \textit{direct} data-driven approach), bypassing the need for an intermediate modeling step~\cite{hou2013model,dorfler2022bridging}. While the first approach leverages the vast array of effective model-based methods for controller synthesis once the identification phase is complete, existing system identification techniques still face significant limitations with respect to the \emph{nonlinear class of unknown systems}~\cite{kerschen2006past}, which can make the process extremely challenging and time-consuming. Even if the precise mathematical model of a system can be identified, these indirect approaches suffer from a \emph{two-layer} complexity: first identifying the model and then solving the problem using model-based methods. In contrast, direct data-driven approaches have a \emph{single-layer} complexity, as they leverage the collected data directly for control analysis.

To provide formal safety controllers for complex dynamical systems with \emph{continuous state spaces}, the existing literature introduces the concept of \textit{control barrier certificates (CBCs)}~\cite{prajna2004safety, prajna2007framework, lavaei2022automated}. In broad terms, a CBC is akin to a \textit{Lyapunov function} that is defined over the state set of the dynamical system, enforcing a set of inequalities on both itself and its one-step ahead alongside the system's discrete-time dynamic. Consequently, if there exists a suitable level set of a CBC that can distinguish an unsafe region from all system trajectories starting from a given initial set, then the CBC's existence offers a formal (probabilistic) guarantee of system safety (see \emph{e.g.,}~\cite{ames2019control,clark2021control,borrmann2015control,nejati2022dissipativity,lavaei2024scalable,nejati2024context,anand2022small,jahanshahi2022compositional,wooding2024protect}, to name a few). A common assumption in all these studies is that the mathematical model of the dynamical system is available. However, in many real-world applications, this precise model is either unavailable or too complex to be practically used, as mentioned earlier.

In recent years, a wide range of data-driven approaches has been proposed to address the challenge of analyzing unknown dynamical systems. Existing studies tackle various problems, including linear quadratic regulation and robust controller design~\cite{de2021low, berberich2020robust, berberich2022combining}, model-reference controller design for linear dynamic systems~\cite{breschi2021direct}, and stabilization of polynomial systems~\cite{guo2021data}, although none of these approaches account for state constraints. More recently, data-driven methods have been introduced for synthesizing state-feedback controllers that ensure a compact polyhedral set containing the origin remains robustly invariant~\cite{bisoffi2020data, bisoffi2022controller}. However, these methods are potentially conservative, as controllers may still exist for subsets of the polytope in cases where no controller is found for the entire specified polyhedral set. Additionally, some approaches leverage \emph{scenario-based programs} to synthesize barrier certificates for the control analysis (\emph{e.g.,}~\cite{nejati2023formal,salamati2024data}). However, these methods require the data to be \textit{independent and identically distributed (i.i.d.)}, which necessitates collecting multiple trajectories—often millions in practical cases—from the unknown system, as opposed to our approach, which requires only a \emph{single input-output trajectory} for the control analysis.

It is worth noting that while there is an extensive body of literature on control analysis for \emph{continuous-time} nonlinear polynomial systems (e.g.,~\cite{bisoffi2020data,bisoffi2022controller,nejati2022data}) using the concept of persistency of excitation introduced by \emph{Willems' fundamental lemma}~\cite{willems2005note}, no work in this area addresses \emph{discrete-time} nonlinear polynomial systems to provide safety certificates using \textit{CBCs}. In particular, while one might assume that providing results for continuous-time systems is more challenging, developing data-driven approaches using the concept of a single trajectory for \textit{discrete-time} unknown systems is actually more difficult (see~\cite{martin2023guarantees}). For instance, the form of CBC for continuous-time nonlinear polynomial systems proposed in~\cite{nejati2022data} cannot be applied to \emph{discrete-time}  systems since the one-step ahead of the proposed function, required in the last condition of the CBC, depends on the one-step-ahead values of the system's monomials, which are not available. More precisely, the work~\cite{nejati2022data} proposes $\mathds{B}(x(k)) = \mathcal{M}^\top(x(k))P \mathcal{M}(x(k))$ as a potential CBC, where $\mathcal{M}(x(k))$ is the system's monomials at time step $k$, and the one-step ahead of this function is $\mathds{B}(x(k+1)) = \mathcal{M}^\top(x(k+1))P \mathcal{M}(x(k+1))$ (cf. Definitions~\ref{def: dt-NPS} and~\ref{def: CBC}). This leads to a dead-end in the context of \emph{discrete-time} nonlinear polynomial systems since the data-based presentation of $\mathcal{M}(x(k+1))$ is unavailable (cf. Lemma~\ref{lem: closed-loop}).

Motivated by this key challenge, this work develops a data-driven method for learning \textit{control barrier certificates} and corresponding \textit{safety controllers} for unknown \textit{discrete-time} nonlinear polynomial systems, ensuring system safety. The approach leverages measured data obtained from input-output observations, referred to as a \textit{single trajectory}, collected during a specified experimental horizon. Under the assumption that the data is persistently excited, satisfying a rank condition, we design control barrier certificates and their associated safety controllers directly from the \textit{finite-length} observed data, through a data-based sum-of-squares optimization program. Unlike the approach in~\cite{nejati2023formal,salamati2024data}, our safety guarantees do not require a large amount of data, nor are they based on probabilistic confidence levels. Additionally, our proposed CBC differs from the form proposed in~\cite{nejati2022data}, as we specifically tackle the more challenging problem in \emph{discrete-time} dynamical systems.

\section{Problem Formulation}\label{sec: problem}
\subsection{Notation}
Sets of real, positive, and non-negative real numbers are denoted by $\mathbb{R},\mathbb{R}^+$, and $\mathbb{R}^+_0$, respectively. Sets of non-negative and positive integers are signified by $\mathbb{N}=\{0,1, 2,\ldots\}$ and $\mathbb{N}^+=\{1,2,\ldots\}$, respectively. An $n \times n$ identity matrix is represented by $\mathds{I}_n$.
Given $N$ vectors $x_i \in \mathbb{R}^{n}$, $x=[x_1 \, \, \ldots \,\, x_N]$ denotes the corresponding matrix of dimension $n \times N$.
We signify that a \emph{symmetric} matrix $P$ is positive definite by $P \succ 0$, while $P \succeq 0$ denotes that $P$ is a \emph{symmetric} positive semi-definite matrix.
The transpose of a matrix $P$ is indicated by $P^\top \!$, while its left pseudoinverse is represented by $P^\dagger$. A star $(\star)$ in a symmetric matrix presents the transposed element in the symmetric position. 

\subsection{Discrete-Time Nonlinear Polynomial Systems}
We begin by introducing the discrete-time nonlinear polynomial systems as the primary model for which we intend to derive a safety certificate.

\begin{definition}[\textbf{dt-NPS}]\label{def: dt-NPS}
	A discrete-time nonlinear polynomial system (dt-NPS) is defined by
	\begin{align}\label{eq: dt-NPS}
		\Sigma\!: x(k+1)=A\mathcal M(x(k)) + Bu(k),
	\end{align}
	where $A \in \mathbb R^{n\times {M}}$ is the system matrix, $B \in \mathbb R^{n\times m}$ is the control matrix, $\mathcal M(x) \in \mathbb R^{M}$ is a vector of monomials in states $x\in  X$, and $u \in U$  is the control input of dt-NPS, with $X \subseteq \mathbb{R}^n$ and $U \subseteq \mathbb{R}^m$ being state and control input sets, respectively. We represent dt-NPS~\eqref{eq: dt-NPS} using the tuple $\Sigma = \left(A, B, X, U\right)$, and denote by $x_{x_0 u}(k)$ the state trajectory of $\Sigma$ achieved at time step $k \in \mathbb{N}$ under a control input $u$ from an initial condition $x_0 = x(0)$.
\end{definition}

In the context of our work, we assume matrices $A$ and $B$ are \emph{unknown}, which mirrors real-world scenarios, and the term \emph{unknown model} is employed to refer to such systems. As the primary goal of this work is to provide a safety certificate for unknown dt-NPS~\eqref{eq: dt-NPS} using the notion of \textit{control barrier certificate}, we formally introduce such a certificate in the following subsection.

\subsection{Control Barrier Certificates}
\begin{definition}[\textbf{CBC}]\label{def: CBC}
	Consider a dt-NPS
	$\Sigma = \left(A, B, X, U\right)\!,$ with $X_0 \subset X$ and $X_u \subset X$ being its initial and unsafe sets, respectively. Assuming the existence of constants $\alpha_1,\alpha_2 \in \mathbb{R}^{+}$, with $\alpha_2 > \alpha_1$, a function $\mathds B:X \to \mathbb{R}_0^+$ is called a control barrier certificate (CBC) for $\Sigma$ if
	\begin{subequations}\label{eq: CBC}
		\begin{align}
			&  \:\:  \mathds B(x) \leq \alpha_1, \hspace{2.55cm} \forall x \in X_{0},\label{subeq: initial}\\
			&  \:\:  \mathds B(x) \geq \alpha_2, \hspace{2.55cm} \forall x \in X_{u}, \label{subeq: unsafe}
		\end{align}  
		and $\forall x\in X$, $\exists u\in U$, such that
		\begin{align}\label{subeq: decreasing}
			\mathds B(x(k+1)) \leq  \mathds B(x(k)).
		\end{align}
	\end{subequations}
\end{definition}

In broad terms, the existence of a CBC for $\Sigma$, as introduced in Definition~\ref{def: CBC},  indicates that the evolution of states never leads to reaching the unsafe set $X_u$ if it is started from any initial conditions within the initial set $X_0$. The following theorem, borrowed from~\cite{prajna2004safety}, leverages the notion of CBC and establishes safety certificates across dt-NPS~\eqref{eq: dt-NPS} within an infinite time horizon.

\begin{theorem}[\textbf{Infinite-horizon safety}]\label{thm: model-based}
	Consider a dt-NPS $\Sigma$ as introduced in Definition~\ref{def: dt-NPS} with $X_0$ and $X_u$ being its initial and unsafe sets, respectively. Suppose $\mathds{B}$ is a CBC for $\Sigma$ as in Definition~\ref{def: CBC}. The dt-NPS is safe in the sense that states' trajectories of dt-NPS never reach the unsafe set $X_u$ starting from the initial set $X_0$, \emph{i.e.,} $x_{x_0 u} \notin X_u$ for any $x_0 \in X_0$ and any $k \in \mathbb{N}$, under the control input $u(\cdot)$ being designed such that~\eqref{subeq: decreasing} holds true.
\end{theorem}

Constructing a CBC and its safety controller for a given dt-NPS $\Sigma$ to ensure system safety, as described in Theorem~\ref{thm: model-based}, requires precise knowledge of the matrices $A$ and $B$, due to their role in~\eqref{subeq: decreasing}. However, in many real-world applications, these matrices are not available, which serves as the primary motivation for this work. Given this significant challenge, we now formally present the problem addressed in our paper.

\begin{resp}
	\begin{problem}
		Consider a dt-NPS $\Sigma$ in Definition~\ref{def: dt-NPS} characterized by \emph{unknown matrices} $A$ and $B$, along with initial and unsafe sets $X_0, X_u \subset X$. By collecting only \emph{one single input-output trajectory} from $\Sigma$, develop a formal data-driven framework upon which a CBC and its corresponding safety controller, as in Definition~\ref{def: CBC}, can be designed using data, thus ensuring the system's safety, as in Theorem~\ref{thm: model-based}.
	\end{problem}
\end{resp}

Having introduced the primary problem aimed at our work, we now proceed with proposing our data-driven framework in the following section.

\section{Data-Driven Construction of CBCs and Safety Controllers}\label{sec: data scheme}
This section is endowed with the data-driven scheme proposed for designing a CBC and its corresponding safety controller for the unknown dt-NPS in Definition~\ref{def: dt-NPS}. To this end, we first propose the structure of the CBC in the quadratic form $\mathds{B}(x) = x^\top P x$, where $P \succ 0$. Subsequently, under the assumption of having all system's states and inputs measured from the unknown dt-NPS~\eqref{eq: dt-NPS}, we collect input-output data (i.e., measurements) within the time horizon $\left[0, \mathcal{T}\right]$, with $\mathcal{T} \in \mathbb{N}^+$ being the number of gathered samples:
\begin{subequations}\label{eq: data}
	\begin{align}
		\mathds{U}_-  & = \left[\begin{array}{cccc}
			\!\! \! u(0) & u(1) & \dots & u(\mathcal{T} - 1) \!\!\!\!
		\end{array}\right]\!,\\
		\mathds{X}_-  & = \left[\begin{array}{cccc}
			\!\! \! x(0) & x(1) & \dots & x(\mathcal{T} - 1) \!\!\!\!
		\end{array}\right]\!,\\
		\mathds{X}_+  & = \left[\begin{array}{cccc}
			\!\! \! x(1) & x(2) & \dots & x(\mathcal{T}) \!\!\!\!
		\end{array}\right]\!.
	\end{align}
\end{subequations}
We recall that the set of input-output trajectories in~\eqref{eq: data} is considered a \textit{single trajectory} of the unknown dt-NPS. 

\begin{remark}
It is worth highlighting that the trajectory of $\mathds{X}_+$ can be precisely collected in the \emph{discrete-time} setting, as the unknown dynamics of the system evolve recursively at discrete time intervals. However, this is not the case in continuous time~\cite{nejati2022data}, where $\mathds{X}_+$ includes derivatives of the state at sampling times, which are generally not available as direct measurements. To tackle this difficulty in continuous-time systems, appropriate filters must be used to approximate these derivatives, while introducing some errors in the approximation process. In contrast, our discrete-time setting provides \emph{exact measurements} of the system’s evolution, representing a key advantage of the  \emph{discrete-time} data-driven approach over continuous-time methods.
\end{remark}

With the inspiration of~\cite{de2019formulas,guo2021data}, we introduce the required conditions that pave the way for characterizing the data-based presentation of dt-NPS in the subsequent lemma.

\begin{lemma}[\textbf{Data-based dt-NPS}]\label{lem: closed-loop}
	Let us suppose $\mathds{Q}(x) \in \mathbb{R}^{\mathcal{T} \times n}$ is a state-dependent matrix such that
	\begin{align}
		\Theta(x) = \mathds{M}_- \mathds{Q}(x), \label{eq: cond_Q}
	\end{align}
	where $\Theta(x) \in \mathbb{R}^{M \times n}$ is a state-dependent \emph{transformation} matrix satisfying
	\begin{align}
		\mathcal{M}(x) = \Theta(x) x, \label{eq: transformation}
	\end{align}
	and $\mathds{M}_-$ is a \emph{full row-rank} matrix formed from the vector $\mathcal{M}(x)$ and samples $\mathds{X}_-$ as
	\begin{align}
		\mathds{M}_- = \left[\begin{array}{cccc}
			\!\! \! \mathcal{M}(x(0)) & \mathcal{M}(x(1)) & \dots & \mathcal{M}(x(\mathcal{T} - 1)) \!\!\!\!
		\end{array}\right]\!\!. \label{eq: M of M}
	\end{align}
	If the control input is synthesized as $u = \digamma(x) x$ with $\digamma(x) = \mathds{U}_- \mathds{Q}(x)$, then the data-based presentation of the closed-loop system is obtained as
	\begin{align}
		A\mathcal M(x) + Bu = (A \Theta(x) + B\digamma(x))x = (\mathds{X}_+ \mathds{Q}(x))x. \label{eq: closed-loop}
	\end{align}
\end{lemma}

\begin{proof}
	Since we set $u = \digamma(x) x$, and considering~\eqref{eq: transformation}, we have the closed-loop dt-NPS  as 
	\begin{align*}
		A\mathcal{M}(x) + Bu = A \Theta(x)x + Bu = (A \Theta(x) + B\digamma(x))x.
	\end{align*}
	Now, according to~\eqref{eq: cond_Q} and since $\digamma(x) = \mathds{U}_- \mathds{Q}(x)$, one has
	\begin{align}\label{new}
		A\mathcal{M}(x) + Bu = (A \mathds{M}_- \mathds{Q}(x) + B \mathds{U}_- \mathds{Q}(x))x= (A \mathds{M}_-  + B \mathds{U}_-)\mathds{Q}(x)x.
	\end{align}
	On the other hand, according to the input-output data collected in~\eqref{eq: data} and matrix $\mathds{M}_-$ in~\eqref{eq: M of M}, we clearly have
	\begin{align}\label{new1}
		\mathds{X}_+ = A \mathds{M}_- + B \mathds{U}_-.
	\end{align}
	Thus from~\eqref{new} and~\eqref{new1}, one can deduce that
	\begin{align*}
		A\mathcal{M}(x) + Bu = (\mathds{X}_+ \mathds{Q}(x))x,
	\end{align*}
	which completes the proof.
\end{proof}

\begin{remark}[\textbf{Rank condition}]
	~In order to have a \emph{full row-rank matrix} $\mathds{M}_-$, one must collect at least $M$ points of data, meaning that $\mathcal{T}$ must be at least equal to $M$. It should be acknowledged that this requirement is readily verifiable since matrix $\mathds{M}_-$ is formed from data.
\end{remark}

\begin{remark}
Note that condition~\eqref{eq: transformation} plays a key role in our approach, as it ultimately translates everything in terms of $x$ rather than $\mathcal M(x)$, which is consistent with the form of our CBC $\mathds{B}(x) = x^\top P x$ and facilitates the proof of our main results in Theorem~\ref{thm: main}. Without loss of generality, there always exists a transformation $\Theta(x)$ that satisfies condition~\eqref{eq: transformation}, allowing us to form $\mathcal{M}(x)$ accordingly.
\end{remark}

Having obtained the dt-NPS's closed-loop data-based representation, as in Lemma~\ref{lem: closed-loop}, we now raise the following theorem as the primary contribution of this work, allowing one to design a CBC and its corresponding safety controller solely on the basis of a \textit{single trajectory} from dt-NPS.

\begin{theorem}[\textbf{Data-driven CBC and safety controller}]\label{thm: main}
	Consider an unknown dt-NPS $\Sigma$, as in Definition~\ref{def: dt-NPS}, with its closed-loop data-based representation $A\mathcal{M}(x) + Bu=(\mathds{X}_+ \mathds{Q}(x))x$, as in Lemma~\ref{lem: closed-loop}. Suppose there exists a state-dependent matrix $\mathds{H}(x) \in \mathbb{R}^{\mathcal{T} \times n}$ such that
	\begin{align}
		\mathds{M}_- \mathds{H}(x) = \Theta(x) P^{-1}.\label{eq: H&P}
	\end{align}
	Now, if there exist constants $\alpha_1, \alpha_2 \in \mathbb{R}^{+}$, with $\alpha_2 > \alpha_1$, such that the following conditions are satisfied
	\begin{subequations}\label{eq: key-conditions}
		\begin{align}
			&x^\top \left[\Theta^\dagger \mathds{M}_- \mathds{H}(x)\right]^{-1} x \leq \alpha_1, \quad\quad \quad \forall x \in X_0, \label{subeq: thm-ini}\\
			&x^\top \left[\Theta^\dagger \mathds{M}_- \mathds{H}(x)\right]^{-1} x \geq \alpha_2, \quad  \quad\quad\forall x \in X_u, \label{subeq: thm-uns}\\
			&\quad~\begin{bmatrix}
				P^{-1} & \mathds{X}_+ \mathds{H}(x)\\
				\star & P^{-1}
			\end{bmatrix} \succeq 0, \qquad\quad\quad\quad\!\!\!\! \forall x \in X,  \label{subeq: thm-dec}
		\end{align}
	\end{subequations}
	then, one can deduce that $\mathds{B}(x) = x^\top \overbrace{\left[\Theta^\dagger \mathds{M}_- \mathds{H}(x)\right]^{-1}}^{P} x$ and $u = \mathds{U}_- \mathds{H}(x) \overbrace{\left[\Theta^\dagger \mathds{M}_- \mathds{H}(x)\right]^{-1}}^{P} x$ are, respectively, a CBC and its corresponding safety controller for the unknown dt-NPS $\Sigma$.
\end{theorem}

\begin{proof}
	We first show the satisfaction of conditions~\eqref{subeq: initial} and~\eqref{subeq: unsafe}. As a result of condition~\eqref{eq: H&P}, one can readily deduce that the satisfaction of conditions~\eqref{subeq: thm-ini} and~\eqref{subeq: thm-uns} implies conditions~\eqref{subeq: initial} and~\eqref{subeq: unsafe} with $P =  \left[\Theta^\dagger \mathds{M}_- \mathds{H}(x)\right]^{-1}$:
	\begin{align*}
		\mathds{M}_- \mathds{H}(x) = \Theta(x) P^{-1} ~\to~~  \Theta^\dagger \mathds{M}_- \mathds{H}(x) = P^{-1} ~\to~~ \left[\Theta^\dagger \mathds{M}_- \mathds{H}(x)\right]^{-1} = P.
	\end{align*}
	 We now proceed with showing the fulfillment of condition~\eqref{subeq: decreasing}, as well. By employing the quadratic definition of CBC, we have
	\begin{align*}
		\mathds{B}(x(k + 1)) = x^\top (k+1) P x(k+1) = (A\mathcal{M}(x) + Bu)^\top P (A\mathcal{M}(x) + Bu).
	\end{align*}
	Then, according to~\eqref{eq: transformation}, and considering $u = \digamma(x)x$, one has
	\begin{align*}
		\mathds{B}(x(k \!+\! 1))  \!=\! (A \Theta(x)x \!+\! B \digamma(x)x)^\top P (A \Theta(x)x \!+\! B \digamma(x)x).
	\end{align*}
	Since according to~\eqref{eq: H&P}, one has $\mathds{M}_- \mathds{H}(x) P = \Theta(x)$, one can leverage~\eqref{eq: cond_Q} and set $\mathds{Q}(x) = \mathds{H}(x) P$, implying $\mathds{Q}(x) P^{-1} = \mathds{H}(x)$.  Since $A \Theta(x) + B \digamma(x) = \mathds{X}_+ \mathds{Q}(x)$ according to Lemma~\ref{lem: closed-loop}, we have
	\begin{align*}
		\mathds{B}(x(k \!+\! 1)) \! &=  x^\top \left[(\mathds{X}_+ \mathds{Q}(x))^\top P (\mathds{X}_+ \mathds{Q}(x))\right] x\\
		&= x^\top \! P^\top \!\! \left[P^{-1}(\mathds{X}_+ \mathds{Q}(x))^\top \! P (\mathds{X}_+ \mathds{Q}(x))P^{-1}\right] \!\! P x\\
		& = x^\top \! P^\top \!\! \left[(\mathds{X}_+ \mathds{H}(x))^\top \! P (\mathds{X}_+ \mathds{H}(x))\right] \!\! P x.
	\end{align*}
	One can also rewrite $\mathds B(x(k)) = x^\top P^\top P^{-1} P x$ since $P$ is symmetric ($P = P^\top$) and $P^{-1} P =\mathds{I}_n$.
	Hence, in order to demonstrate that $\mathds B(x(k+1)) \leq  \mathds B(x(k))$, it suffices to show
	\begin{align*}
		 (\mathds{X}_+ \mathds{H}(x))^\top  P (\mathds{X}_+ \mathds{H}(x)) \preceq P^{-1}
		\to~ P^{-1} - (\mathds{X}_+ \mathds{H}(x))^\top  P (\mathds{X}_+ \mathds{H}(x)) \succeq 0.
	\end{align*}
	According to the Schur complement~\cite{zhang2006schur} and condition~\eqref{subeq: thm-dec}, one has
	\begin{align*}
		P^{-1} \!-\! (\mathds{X}_+ \mathds{H}(x))^\top  \! P (\mathds{X}_+ \mathds{H}(x)) \! \succeq \! 0  \Leftrightarrow \begin{bmatrix}
			P^{-1} \!\!&\!\! \mathds{X}_+ \mathds{H}(x)\\
			\star \!\!&\!\! P^{-1}
		\end{bmatrix} \!\! \succeq \! 0, 
	\end{align*}
	resulting in fulfilling condition~\eqref{subeq: decreasing}. Hence, $\mathds{B}(x) = x^\top  \left[\Theta^\dagger \mathds{M}_- \mathds{H}(x)\right]^{-1}  x$ and $u = \mathds{U}_- \mathds{H}(x) \left[\Theta^\dagger \mathds{M}_- \mathds{H}(x)\right]^{-1} x$ are a CBC and its safety controller, respectively, thus concluding the proof.
\end{proof}

Having offered the main theorem of this work, in the following subsection, we propose a method upon which CBCs and corresponding safety controllers can be computed systematically.

\subsection{Computation of CBCs and Safety Controllers}
In this subsection, we first present the following lemma, which transforms conditions~\eqref{subeq: thm-ini}-\eqref{subeq: thm-dec} into a sum-of-squares (SOS) optimization program for computing CBCs and corresponding safety controllers. We then introduce an algorithm that outlines the necessary steps for synthesizing them.

\begin{lemma}[\textbf{SOS optimization program}]
	~Consider the state set $X$, the initial set $X_0$, and the unsafe set $X_u$, each of which is outlined by vectors of polynomial inequalities as $X = \{x \in \mathbb{R}^{n} |~ \beth(x) \geq 0\}$, $X_{0} = \{x \in \mathbb{R}^{n} |~ \beth_{0}(x) \geq 0\}$, and $X_{u} = \{x \in \mathbb{R}^{n} |~ \beth_{u}(x) \geq 0\}$, respectively. Then, $\mathds{B}(x) = x^\top  \left[\Theta^\dagger \mathds{M}_- \mathds{H}(x)\right]^{-1}  x$ is a CBC for the unknown dt-NPS~\eqref{eq: dt-NPS}, fulfilling all conditions in~\eqref{eq: key-conditions}, and $u = \mathds{U}_- \mathds{H}(x) \left[\Theta^\dagger \mathds{M}_- \mathds{H}(x)\right]^{-1} x$ is its safety controller if there exist a state-dependent-polynomial matrix $\mathds{H}(x) \in \mathbb{R}^{\mathcal{T} \times n}$, constants $\alpha_1, \alpha_2 \in \mathbb{R}^{+}$, with $\alpha_2 > \alpha_1$, and vectors of sum-of-squares polynomials  $\Lambda(x), \Lambda_0(x), \Lambda_u(x)$, such that
	\begin{subequations}\label{eq: SOS conditions}
		\begin{align}
			-& x^\top \left[\Theta^\dagger \mathds{M}_- \mathds{H}(x)\right]^{-1} x - \Lambda_0^\top(x) \beth_{0}(x) + \alpha_1,\label{subeq: lem-ini}\\
			& x^\top \left[\Theta^\dagger \mathds{M}_- \mathds{H}(x)\right]^{-1} x - \Lambda_u^\top(x) \beth_{u}(x) - \alpha_2,\label{subeq: lem-uns}\\
			& \begin{bmatrix}
				P^{-1} & \mathds{X}_+ \mathds{H}(x)\\
				\star & P^{-1}
			\end{bmatrix} - (\Lambda^\top(x) \beth(x))\mathds{I}_{2n}, \label{subeq: lem-dec}
		\end{align}
		are all SOS polynomials, while condition~\eqref{eq: H&P} is also satisfied.
	\end{subequations}
\end{lemma}

\begin{proof}
	Having known that $\Lambda_0(x)$ is an SOS polynomial, we can deduce that $\Lambda_0^\top(x) \beth_{0}(x) \geq 0$ within $X_{0} = \{x \in \mathbb{R}^{n} |~ \beth_{0}(x) \geq 0\}$. Given that $x^\top \left[\Theta^\dagger \mathds{M}_- \mathds{H}(x)\right]^{-1} x$, with $ \left[\Theta^\dagger \mathds{M}_- \mathds{H}(x)\right]^{-1} = P \succ 0$, is a non-negative SOS polynomial, one can conclude that the satisfaction of~\eqref{subeq: lem-ini} implies the satisfaction of~\eqref{subeq: thm-ini}. Likewise, one can conclude this reasoning between~\eqref{subeq: lem-uns} and~\eqref{subeq: thm-uns}. We shall now continue with showing condition~\eqref{subeq: thm-dec}, as well. Since $\Lambda(x)$ is an SOS polynomial, one can deduce that $\Lambda^\top(x) \beth(x) \geq 0$ within $X = \{x \in \mathbb{R}^{n} |~ \beth(x) \geq 0\}$. Given that~\eqref{subeq: lem-dec} is also an SOS polynomial, we have
	\begin{align*}
		\begin{bmatrix}
			P^{-1} & \mathds{X}_+ \mathds{H}(x)\\
			\star & P^{-1}
		\end{bmatrix} - (\Lambda^\top(x) \beth(x))\mathds{I}_{2n} \succeq 0,
	\end{align*}
	deducing that the satisfaction of~\eqref{subeq: lem-dec} implies the satisfaction of~\eqref{subeq: thm-dec}, thus completing the proof.
\end{proof}

We present Algorithm~\ref{alg}, in which all the required steps of designing CBCs and safety controllers are outlined.

\begin{algorithm}[t!]
	\caption{Data-driven construction of CBCs and corresponding safety controllers}\label{alg}
	\begin{algorithmic}[1]
		\REQUIRE The state set $X$, the initial set $X_0$, the unsafe set $X_u$, and the choice of $\mathcal{M}(x)$\footnotemark
		\STATE Collect $\mathds{U}_-, \mathds{X}_-, \mathds{X}_+$ as in~\eqref{eq: data}
		\STATE Form $\Theta(x)$ and $\mathds{M}_-$ according to~\eqref{eq: transformation} and~\eqref{eq: M of M}
		\STATE Utilize \textsf{SOSTOOLS} to obtain $P$ and $\mathds{H}(x)$ that fulfill conditions~\eqref{subeq: lem-dec} and~\eqref{eq: H&P} simultaneously (see Remark~\ref{remark: invP} for more details)\label{step3}
		\STATE Given the constructed $\mathds{B}(x) = x^\top P x$, utilize  \textsf{SOSTOOLS} to design $\alpha_1$ and $\alpha_2$ as in~\eqref{subeq: lem-ini} and~\eqref{subeq: lem-uns}, respectively, where $\alpha_2 > \alpha_1$
		\ENSURE CBC $\mathds{B}(x)=x^\top P x$ and its corresponding safety controller $u = \mathds{U}_- \mathds{H}(x) P x$
	\end{algorithmic}
\end{algorithm}
\footnotetext{If an upper bound on the maximum degree of $\mathcal{M}(x)$ can be inferred based on physical insight of unknown system, $\mathcal{M}(x)$ can be selected to include all possible combinations of states up to that known upper bound.}

\begin{remark}[\textbf{Enforcement of SOS conditions}]\label{remark: tools}
	To effectively enforce conditions~\eqref{subeq: lem-ini}-\eqref{subeq: lem-dec}, one can utilize existing software tools such as \textsf{SOSTOOLS}~\cite{prajna2004sostools} with semi-definite programming (SDP) solvers like \textsf{SeDuMi}~\cite{sturm1999using}.
\end{remark}

\begin{remark}[\textbf{On Step~\ref{step3} of Algorithm~\ref{alg}}]\label{remark: invP}
	To fulfill condition~\eqref{subeq: lem-dec}, we define $\mathcal{Z} = P^{-1}$ while enforcing it to be a \emph{symmetric positive-definite} matrix, i.e., $\mathcal{Z} \succ 0$. After satisfying condition~\eqref{subeq: lem-dec} and designing $\mathcal{Z}$, the matrix $P$ is obtained as $\mathcal{Z}^{-1} = (P^{-1})^{-1} = P$. 
\end{remark}

\section{Simulation Results}\label{sec: simul}
We showcase the effectiveness of our proposed findings by implementing them on two physical case studies including a jet engine~\cite{anta2010sample} and a Lorenz system~\cite{lopez2019synchronization}, with \emph{unknown matrices} $A$ and $B$. The data collected for each case study together with matrices $\mathds{M}_-$ and $\Theta(x)$ are reported in Appendix. It is worth noting that all simulations were performed on a MacBook with an M2 chip and 32 GB of memory.

\textbf{Jet engine.} Regions of interest are given as $X = [-10, 10]^2$, $X_0 = [0, 2] \times [-2, 2]$, $X_u = [-5, -2.5]^2 \cup [2.5, 5]^2$. The key objective is to design a CBC and its safety controller for a jet engine with unknown model, while ensuring that the system's states remain within the comfort zone of $X \backslash X_u$ over an infinite time horizon. We select $\mathcal{M}(x)= \begin{bmatrix}
    x_1 & \! x_2 & \! x^2_1 & \! x_1  x_2 & \! x^2_2 & \! x^3_1 & \! x^2_1  x_2 & \! x_1  x^2_2 & \! x^3_2
\end{bmatrix}^\top$ and follow the first step of Algorithm~\ref{alg} to collect input-output trajectories $\mathds{U}_-,\mathds{X}_-,  \mathds{X}_+$, as specified in~\eqref{eq: data} over the time horizon of $\mathcal{T}=15$. 
 Subsequently, we compute $\Theta(x)$ and $\mathds{M}_-$ according to~\eqref{eq: transformation} and~\eqref{eq: M of M}, respectively. Using \textsf{SeDuMi} alongside with \textsf{SOSTOOLS}, we first obtain $\mathds{H}(x)$ and $P$ as
 \begin{align}\label{Martix-P}
    P = 10^{4}\times\begin{bmatrix}
           4.8273  &  0.0023\\
    0.0023  &  0.0161 
     \end{bmatrix}\!\!.
 \end{align}
Then, we compute the CBC using $P$ in \eqref{Martix-P}  as
\begin{align}\label{Barrier}
  \mathds{B}(x) \!=\!  48272.6605 x_1^2 \!+\! 46.9317 x_1 x_2 \!+\! 161.1994 x_2^2.
\end{align}
Based on the constructed \( \mathds{B}(x) \)  in \eqref{Barrier}, and satisfying conditions in~\eqref{subeq: lem-ini} and~\eqref{subeq: lem-uns}, we design \(\alpha_1 =  1.9392 \times 10^{5}\) and \(\alpha_2 = 3.03 \times 10^{5}\), where \( \alpha_2 > \alpha_1 \). A safety controller is also designed as 
\begin{align}\nonumber
u=&\:0.00162 x_1^3 - 0.0020741 x_1^2 x_2 + 0.00010883 x_1 x_2^2 - 8.9663 \times 10^{-6} x_2^3 - 0.095569 x_1^2\\\label{controller} & - 0.51382 x_1 x_2 - 0.0026997 x_2^2 + 88.3896 x_1 + 1000.0175 x_2.  
\end{align}

By leveraging Theorem~\ref{thm: model-based}, we ensure that all trajectories of the unknown jet engine, starting from \( X_0 = [0, 2] \times [-2, 2] \), remain within the safe set \(X \backslash X_u\) over an infinite time horizon. Fig. \ref{fig:b} displays the state trajectories of the unknown jet engine when the controller in~\eqref{controller} is applied, confirming compliance with the desired safety specification.

\begin{figure}[t!]
	\centering
	\includegraphics[width=0.4\linewidth]{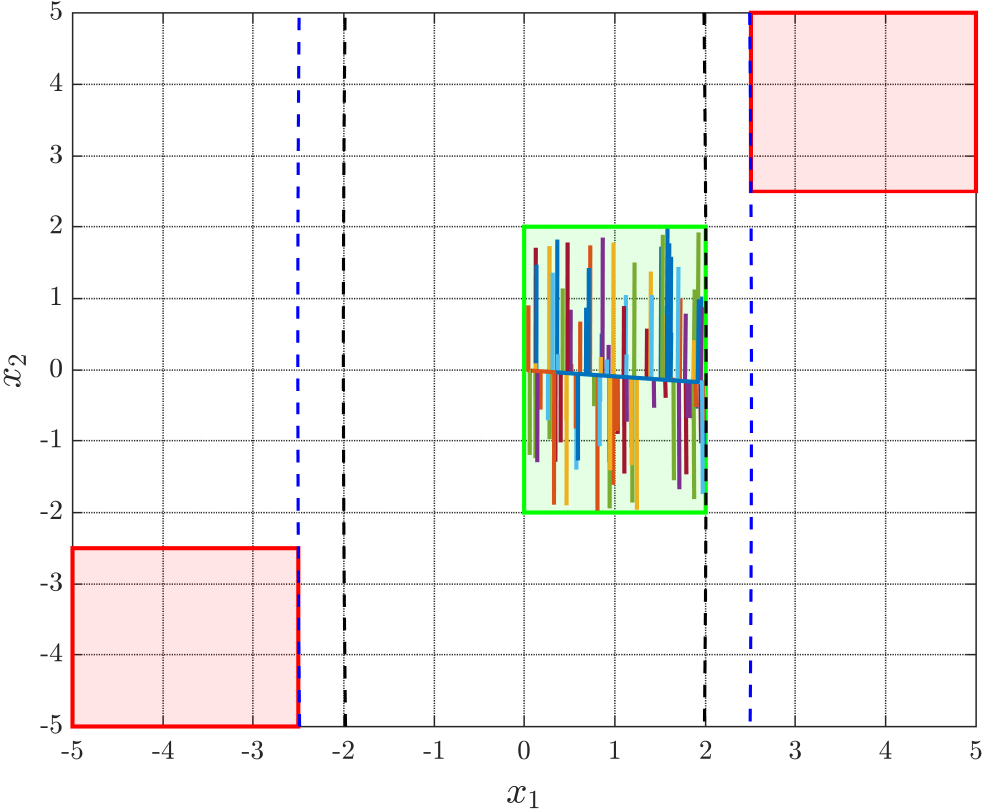}
	\caption{$100$ state trajectories of unknown jet engine with the designed controller in~\eqref{controller} starting from different initial conditions in~$X_0 \in  [0, 2] \times [-2, 2] $. Initial and unsafe regions are depicted by green \protect\greensquare\ and red \protect\redsquare\ boxes, while $\mathds B(x) = \alpha_1$ and $\mathds B(x) = \alpha_2$ are indicated by~\sampleline{dashed, black, thick} and~\sampleline{dashed, blue, thick}.}
	\label{fig:b}
\end{figure}

\textbf{Lorenz system.} A well-known example of a chaotic dynamical system is the \emph{Lorenz system}, which is often used in studies of complex behavior in nonlinear systems.  For this case study, the regions of interest are given as $X = [-12, 12]^3$, $X_0 = [0, 2] \times [-2, 2]^2$, and $X_u =  [-5, -2.5]^3 \cup [2.5, 5]^3$. To begin, we select the monomials as $\mathcal{M}(x) = \begin{bmatrix} x_1 & x_2 & x_3 & x_1x_2 & x_2x_3 & x_1x_3 \end{bmatrix}^\top$ and gather input-output trajectories $ \mathds{U}_-, \mathds{X}_-,\mathds{X}_+$ over a time horizon $\mathcal{T} = 15$. 

Next, we calculate $\Theta(x)$ and $\mathds{M}_-$ using~\eqref{eq: transformation} and~\eqref{eq: M of M}, respectively. By utilizing \textsf{SeDuMi} in conjunction with \textsf{SOSTOOLS}, we derive the matrices $\mathds{H}(x)$ and $P$, with the matrix $P$ reported as
\begin{align}\label{eq:P Lorenz}
    P = 10^{4}\times\begin{bmatrix}
           0.0636   &  -0.0343 & 0.0209\\
    -0.0343  &  0.1214  &  -0.0003\\
    0.0209  &  -0.0003  & 8.6481
     \end{bmatrix}\!\!.
\end{align}
Subsequently, the CBC is computed using $P$ in~\eqref{eq:P Lorenz} as
\begin{align}\label{eq:Barrier Lorenz}
    \mathds{B}(x)  &=636.2337 x_1^2 - 686.577 x_1 x_2 + 417.4737 x_1 x_3 + 1214.2754 x_2^2 - 5.3721 x_2 x_3 + 86480.554 x_3^2.
\end{align}

Based on the constructed $\mathds{B}(x)$ in~\eqref{eq:Barrier Lorenz}, and by fulfilling conditions~\eqref{subeq: lem-ini} and~\eqref{subeq: lem-uns} using \textsf{SOSTOOLS}, we compute the level sets $\alpha_1 =  3.5776 \times 10^{5}$ and $\alpha_2 = 5.5035 \times 10^{5}$, where $\alpha_2 > \alpha_1$. 

A safety controller is also synthesized as follows:
\begin{align}\label{eq:controller Lorenz}
    u = ~\!& 1.7859 \times 10^{-6} x_1^3 - 1.9979 \times 10^{-6} x_1^2 x_2 - 0.00010195 x_1^2 x_3 - 1.3839 \times 10^{-6} x_1 x_2^2 \notag\\
    &+ 0.00020084 x_1 x_2 x_3 - 0.00017269 x_1 x_3^2 \notag - 3.0862 \times 10^{-7} x_2^3 + 3.5187 \times 10^{-5} x_2^2 x_3 \notag\\
    &- 1.6423 \times 10^{-5} x_2 x_3^2 + 2.7836 \times 10^{-5} x_3^3 - 0.010953 x_1^2 + 1.0254 x_1 x_2 - 9.975 x_1 x_3 \notag\\
    &+ 0.0086013 x_2^2 - 0.39793 x_2 x_3 + 0.34484 x_3^2 + 252.5353 x_1 - 988.7662 x_2 + 23.0669 x_3.
\end{align}

\begin{figure}[t!]
	\centering
	\includegraphics[width=0.48\linewidth]{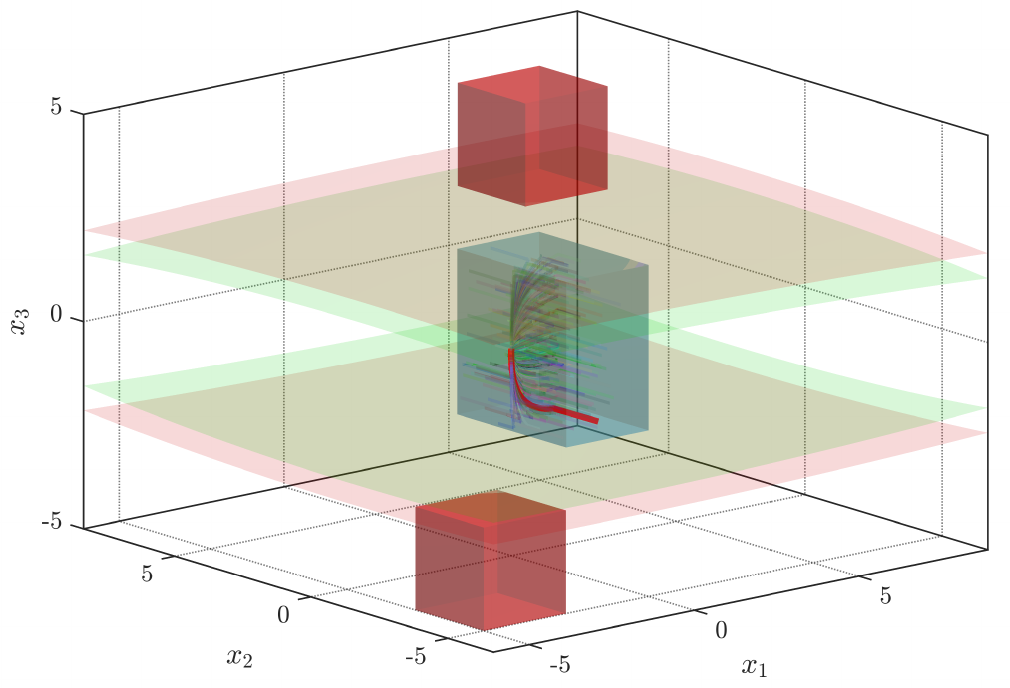}
	\caption{$250$ state trajectories of unknown Lorenz system and a representative trajectory \sampleline{red, very thick} under the designed controller in~\eqref{eq:controller Lorenz}, starting from different initial conditions in~$X_0 \in  [0, 2] \times [-2, 2]^2$. Initial and unsafe regions are depicted by blue \protect\bluesquare\ and red \protect\reddsquare\ boxes, while $\mathds B(x) = \alpha_1$ and $\mathds B(x) = \alpha_2$ are indicated by~\protect\greensquare\ and~\protect\redsquare\,, respectively.}
	\label{fig:b1}
\end{figure}

Using Theorem~\ref{thm: model-based}, we verify that all trajectories of the unknown Lorenz system, starting from the initial set $X_0$, remain within the safe set $X \backslash X_u$ over an infinite time horizon. Closed-loop state trajectories of the Lorenz system under the control law~\eqref{eq:controller Lorenz} are illustrated in Fig.~\ref{fig:b1}, demonstrating adherence to the specified safety criteria.

\section{Conclusion}\label{sec: conc}
In this work, we developed a \emph{direct} data-driven method for learning control barrier certificates (CBCs) and safety controllers for unknown discrete-time nonlinear polynomial systems, ensuring system safety over an infinite time horizon. By leveraging measured data from a single input-output trajectory, we were able to design a CBC and its corresponding safety controller directly from the finite-length observed data. The effectiveness of our approach was demonstrated through two case studies, including a jet engine and a chaotic Lorenz system. Developing a data-driven approach that encompasses \emph{more general discrete-time nonlinear systems} beyond polynomials is currently under investigation as future work.

\bibliographystyle{alpha}
\bibliography{biblio}

\newcommand{\etalchar}[1]{$^{#1}$}
\begin{thebibliography}{LMLCPC{\etalchar{+}}19}

\bibitem[ACE{\etalchar{+}}19]{ames2019control}
A.~D. Ames, S.~Coogan, M.~Egerstedt, G.~Notomista, K.~Sreenath, and P.~Tabuada.
\newblock Control barrier functions: {T}heory and applications.
\newblock In {\em Proceedings of the 18th European control conference (ECC)},
  pages 3420--3431. IEEE, 2019.

\bibitem[ALZ22]{anand2022small}
M.~Anand, A.~Lavaei, and M.~Zamani.
\newblock From small-gain theory to compositional construction of barrier
  certificates for large-scale stochastic systems.
\newblock {\em IEEE Transactions on Automatic Control}, 67(10):5638--5645,
  2022.

\bibitem[AT10]{anta2010sample}
A.~Anta and P.~Tabuada.
\newblock To sample or not to sample: {S}elf-triggered control for nonlinear
  systems.
\newblock {\em IEEE Transactions on Automatic Control}, 55(9):2030--2042, 2010.

\bibitem[BDPFT21]{breschi2021direct}
V.~Breschi, C.~De~Persis, S.~Formentin, and P.~Tesi.
\newblock Direct data-driven model-reference control with {L}yapunov stability
  guarantees.
\newblock In {\em Proceedings of the 60th Conference on Decision and Control
  (CDC)}, pages 1456--1461. IEEE, 2021.

\bibitem[BDPT20]{bisoffi2020data}
A.~Bisoffi, C.~De~Persis, and P.~Tesi.
\newblock Data-based guarantees of set invariance properties.
\newblock {\em IFAC-PapersOnLine}, 53(2):3953--3958, 2020.

\bibitem[BDPT22]{bisoffi2022controller}
A.~Bisoffi, C.~De~Persis, and P.~Tesi.
\newblock Controller design for robust invariance from noisy data.
\newblock {\em IEEE Transactions on Automatic Control}, 68(1):636--643, 2022.

\bibitem[BKSA20]{berberich2020robust}
J.~Berberich, A.~Koch, C.~W. Scherer, and F.~Allg{\"o}wer.
\newblock Robust data-driven state-feedback design.
\newblock In {\em Proceedings of American Control Conference (ACC)}, pages
  1532--1538. IEEE, 2020.

\bibitem[BSA22]{berberich2022combining}
J.~Berberich, C.~W. Scherer, and F.~Allg{\"o}wer.
\newblock Combining prior knowledge and data for robust controller design.
\newblock {\em IEEE Transactions on Automatic Control}, 68(8):4618--4633, 2022.

\bibitem[BWAE15]{borrmann2015control}
U.~Borrmann, L.~Wang, A.~D. Ames, and M.~Egerstedt.
\newblock Control barrier certificates for safe swarm behavior.
\newblock {\em IFAC-PapersOnLine}, 48(27):68--73, 2015.

\bibitem[Cla21]{clark2021control}
A.~Clark.
\newblock Control barrier functions for stochastic systems.
\newblock {\em Automatica}, 130, 2021.

\bibitem[DCM22]{dorfler2022bridging}
F.~D{\"o}rfler, J.~Coulson, and I.~Markovsky.
\newblock Bridging direct and indirect data-driven control formulations via
  regularizations and relaxations.
\newblock {\em IEEE Transactions on Automatic Control}, 68(2):883--897, 2022.

\bibitem[DPT19]{de2019formulas}
C.~De~Persis and P.~Tesi.
\newblock Formulas for data-driven control: Stabilization, optimality, and
  robustness.
\newblock {\em IEEE Transactions on Automatic Control}, 65(3):909--924, 2019.

\bibitem[DPT21]{de2021low}
C.~De~Persis and P.~Tesi.
\newblock Low-complexity learning of linear quadratic regulators from noisy
  data.
\newblock {\em Automatica}, 128, 2021.

\bibitem[GDPT21]{guo2021data}
M.~Guo, C.~De~Persis, and P.~Tesi.
\newblock Data-driven stabilization of nonlinear polynomial systems with noisy
  data.
\newblock {\em IEEE Transactions on Automatic Control}, 67(8):4210--4217, 2021.

\bibitem[HW13]{hou2013model}
Z.~S. Hou and Z.~Wang.
\newblock From model-based control to data-driven control: Survey,
  classification and perspective.
\newblock {\em Information Sciences}, 235:3--35, 2013.

\bibitem[JLZ22]{jahanshahi2022compositional}
N.~Jahanshahi, A.~Lavaei, and M.~Zamani.
\newblock Compositional construction of safety controllers for networks of
  continuous-space {POMDPs}.
\newblock {\em IEEE Transactions on Control of Network Systems}, 10(1):87--99,
  2022.

\bibitem[KWVG06]{kerschen2006past}
G.~Kerschen, K.~Worden, A.~F. Vakakis, and J.~C. Golinval.
\newblock Past, present and future of nonlinear system identification in
  structural dynamics.
\newblock {\em Mechanical Systems and Signal Processing}, 20(3):505--592, 2006.

\bibitem[LF24]{lavaei2024scalable}
A.~Lavaei and E.~Frazzoli.
\newblock Scalable synthesis of safety barrier certificates for networks of
  stochastic switched systems.
\newblock {\em IEEE Transactions on Automatic Control}, 2024.

\bibitem[LMLCPC{\etalchar{+}}19]{lopez2019synchronization}
D.~L{\'o}pez-Mancilla, G.~L{\'o}pez-Cahuich, C.~Posadas-Castillo, C.~E.
  Casta{\~n}eda, J.~H. Garc{\'\i}a-L{\'o}pez, J.~L. V{\'a}zquez-Guti{\'e}rrez,
  and E.~Tlelo-Cuautle.
\newblock Synchronization of complex networks of identical and nonidentical
  chaotic systems via model-matching control.
\newblock {\em Plos one}, 14(5), 2019.

\bibitem[LSAZ22]{lavaei2022automated}
A.~Lavaei, S.~Soudjani, A.~Abate, and M.~Zamani.
\newblock Automated verification and synthesis of stochastic hybrid systems:
  {A} survey.
\newblock {\em Automatica}, 146, 2022.

\bibitem[MSA23]{martin2023guarantees}
T.~Martin, T.~B. Sch{\"o}n, and F.~Allg{\"o}wer.
\newblock Guarantees for data-driven control of nonlinear systems using
  semidefinite programming: {A} survey.
\newblock {\em Annual Reviews in Control}, 2023.

\bibitem[NLJ{\etalchar{+}}23]{nejati2023formal}
A.~Nejati, A.~Lavaei, P.~Jagtap, S.~Soudjani, and M.~Zamani.
\newblock Formal verification of unknown discrete-and continuous-time systems:
  {A} data-driven approach.
\newblock {\em IEEE Transactions on Automatic Control}, 68(5):3011--3024, 2023.

\bibitem[NPNS24]{nejati2024context}
A.~Nejati, S.~Prakash~Nayak, and A.-K. Schmuck.
\newblock Context-triggered games for reactive synthesis over stochastic
  systems via control barrier certificates.
\newblock In {\em Proceedings of the 27th ACM International Conference on
  Hybrid Systems: Computation and Control}, pages 1--12, 2024.

\bibitem[NZ22]{nejati2022dissipativity}
A.~Nejati and M.~Zamani.
\newblock From dissipativity theory to compositional construction of control
  barrier certificates.
\newblock {\em Leibniz Transactions on Embedded Systems (LITES) (Special Issue
  on Distributed Hybrid Systems)}, 8(2), 2022.

\bibitem[NZCZ22]{nejati2022data}
A.~Nejati, B.~Zhong, M.~Caccamo, and M.~Zamani.
\newblock Data-driven controller synthesis of unknown nonlinear polynomial
  systems via control barrier certificates.
\newblock In {\em Proceedings of Learning for Dynamics and Control Conference},
  pages 763--776. PMLR, 2022.

\bibitem[PJ04]{prajna2004safety}
S.~Prajna and A.~Jadbabaie.
\newblock Safety verification of hybrid systems using barrier certificates.
\newblock In {\em Proceedings of International Workshop on Hybrid Systems:
  Computation and Control}, pages 477--492. Springer, 2004.

\bibitem[PJP07]{prajna2007framework}
S.~Prajna, A.~Jadbabaie, and G.~J. Pappas.
\newblock A framework for worst-case and stochastic safety verification using
  barrier certificates.
\newblock {\em IEEE Transactions on Automatic Control}, 52(8):1415--1428, 2007.

\bibitem[PPSP04]{prajna2004sostools}
S.~Prajna, A.~Papachristodoulou, P.~Seiler, and P.~A. Parrilo.
\newblock {SOSTOOLS}: Control applications and new developments.
\newblock In {\em Proceedings of IEEE International Conference on Robotics and
  Automation}, pages 315--320, 2004.

\bibitem[SLSZ24]{salamati2024data}
A.~Salamati, A.~Lavaei, S.~Soudjani, and M.~Zamani.
\newblock Data-driven verification and synthesis of stochastic systems via
  barrier certificates.
\newblock {\em Automatica}, 159, 2024.

\bibitem[Stu99]{sturm1999using}
J.~F. Sturm.
\newblock Using {S}e{D}u{M}i 1.02, a {MATLAB} toolbox for optimization over
  symmetric cones.
\newblock {\em Optimization Methods and Software}, 11(1-4):625--653, 1999.

\bibitem[WHL24]{wooding2024protect}
Ben Wooding, Viacheslav Horbanov, and Abolfazl Lavaei.
\newblock {PRoTECT}: Parallelized construction of safety barrier certificates
  for nonlinear polynomial systems.
\newblock {\em arXiv:2404.14804}, 2024.

\bibitem[WRMDM05]{willems2005note}
J.~C. Willems, P.~Rapisarda, I.~Markovsky, and B.~L.~M. De~Moor.
\newblock A note on persistency of excitation.
\newblock {\em Systems \& Control Letters}, 54(4):325--329, 2005.

\bibitem[Zha06]{zhang2006schur}
F.~Zhang.
\newblock {\em The Schur complement and its applications}, volume~4.
\newblock Springer Science \& Business Media, 2006.

\end{thebibliography}

\newpage

\section{Appendix: Collected Data and Designed Matrices}

\subsection{Jet engine}
\quad\\\hspace*{2em}\textbf{\emph{Collected data.}}
\small\begin{align*}
	\mathds{U}_-=&\left[\begin{array}{ccccccc} 
		93.41 & 9.446 & 94.54 & 42.96 & 39.55 & -56.78 & 95.25 \\
	\end{array}\right.\\
	&\left.\begin{array}{cccccccc} 
		-98.75 & -49.4 & -13.04 & 55.88 & -60.46 & 72.6 & 96.68 & -67.23 
	\end{array}\right]\\
	\mathds{X}_- = &\left[\begin{array}{cccccccc} 
		0.025 & 0.02498 & 0.02505 & 0.02513 & 0.02531 & 0.02553 & 0.02579 & 0.02599 \\
		0.02 & -0.07338 & -0.0828 & -0.1773 & -0.2203 & -0.2598 & -0.203 & -0.2982 \\
	\end{array}\right.\\
	&\left.\begin{array}{ccccccc} 
		0.02629 & 0.02649 & 0.02663 & 0.02677 & 0.02696 & 0.02709 & 0.0273 \\
		-0.1994 & -0.15 & -0.1369 & -0.1928 & -0.1323 & -0.2048 & -0.3015 \\
	\end{array}\right]\\
	\mathds{X}_+ = &\left[\begin{array}{cccccccc} 
		0.02498 & 0.02505 & 0.02513 & 0.02531 & 0.02553 & 0.02579 & 0.02599 & 0.02629 \\
		-0.07338 & -0.0828 & -0.1773 & -0.2203 & -0.2598 & -0.203 & -0.2982 & -0.1994 \\
	\end{array}\right.\\
	&\left.\begin{array}{ccccccc} 
		0.02649 & 0.02663 & 0.02677 & 0.02696 & 0.02709 & 0.0273 & 0.0276 \\
		-0.15 & -0.1369 & -0.1928 & -0.1323 & -0.2048 & -0.3015 & -0.2342 \\
	\end{array}\right]\\
	\mathds{M}_- = &\left[\begin{array}{cccccc} 
		0.025 & 0.02498 & 0.02505 & 0.02513 & 0.02531 & 0.02553 \\
		0.02 & -0.07338 & -0.0828 & -0.1773 & -0.2203 & -0.2598 \\
		0.000625 & 0.000624 & 0.0006276 & 0.0006317 & 0.0006406 & 0.0006517 \\
		0.0005 & -0.001833 & -0.002074 & -0.004457 & -0.005575 & -0.006632 \\
		0.0004 & 0.005385 & 0.006856 & 0.03144 & 0.04851 & 0.06748 \\
		1.563e-5 & 1.559e-5 & 1.572e-5 & 1.588e-5 & 1.621e-5 & 1.664e-5 \\
		1.25e-5 & -4.579e-5 & -5.196e-5 & -0.000112 & -0.0001411 & -0.0001693 \\
		1.0e-5 & 0.0001345 & 0.0001718 & 0.0007902 & 0.001228 & 0.001723 \\
		8.0e-6 & -0.0003951 & -0.0005677 & -0.005575 & -0.01068 & -0.01753 \\
	\end{array}\right.\\
	&\left.\quad\begin{array}{cccccc} 
		0.02579 & 0.02599 & 0.02629 & 0.02649 & 0.02663 \\
		-0.203 & -0.2982 & -0.1994 & -0.15 & -0.1369 \\
		0.000665 & 0.0006755 & 0.000691 & 0.0007015 & 0.0007094 \\
		-0.005234 & -0.00775 & -0.005242 & -0.003972 & -0.003647 \\
		0.04119 & 0.08892 & 0.03977 & 0.0225 & 0.01875 \\
		1.715e-5 & 1.756e-5 & 1.816e-5 & 1.858e-5 & 1.889e-5 \\
		-0.000135 & -0.0002014 & -0.0001378 & -0.0001052 & -9.713e-5 \\
		0.001062 & 0.002311 & 0.001045 & 0.0005958 & 0.0004993 \\
		-0.008361 & -0.02652 & -0.00793 & -0.003374 & -0.002567 \\
	\end{array}\right.\\
\end{align*}
\small\begin{align*}
	&\left.\quad\quad\quad\quad\quad\quad\begin{array}{cccccc} 
		0.02677 & 0.02696 & 0.02709 & 0.0273 \\
		-0.1928 & -0.1323 & -0.2048 & -0.3015 \\
		0.0007166 & 0.0007269 & 0.000734 & 0.0007451 \\
		-0.00516 & -0.003566 & -0.00555 & -0.00823 \\
		0.03716 & 0.0175 & 0.04196 & 0.0909 \\
		1.918e-5 & 1.96e-5 & 1.989e-5 & 2.034e-5 \\
		-0.0001381 & -9.616e-5 & -0.0001504 & -0.0002247 \\
		0.0009947 & 0.0004718 & 0.001137 & 0.002481 \\
		-0.007163 & -0.002314 & -0.008596 & -0.02741 \\
	\end{array}\right]
\end{align*}
\quad\\\hspace*{2em}{\textbf{\emph{Transformation matrix $	\Theta$, and matrix $P$ designed by} \textsf{SOSTOOLS}.}}\footnote{Since we collected $15$ data points from the unknown jet engine system, we omit reporting $\mathds{H}(x)$ here for brevity.}
\begin{align*}
	\Theta(x) = \begin{bmatrix}
		1 & 0 & x_1 & 0 & 0 & x_1^2 & x_1 x_2 & 0 & 0 \\
		0 & 1 & 0 & x_1 & x_2 & 0 & 0 & x_1 x_2 & x_2^2
	\end{bmatrix}^\top\!\!\!\!\!\!,\quad \quad
	\quad P=\left[\begin{array}{cc} 48270.0 & 23.47\\ 23.47 & 161.2 \end{array}\right]\\
\end{align*}    

\subsection{Lorenz}
\quad\\\hspace*{2em}\textbf{\emph{Collected data.}}
\small\begin{align*}
	\mathds{U}_- = &\left[\begin{array}{cccccccccc}
		93.41 & 9.446 & 94.54 & 42.96 & 39.55 & -56.78 & 95.25 & -98.75 & -49.4 & -13.04
	\end{array}\right.\\
	&\left.\begin{array}{ccccc}
		55.88 & -60.46 & 72.6 & 96.68 & -67.23
	\end{array}\right]\\
	\mathds{X}_- = &\left[\begin{array}{cccccccccc}
		1.5 & 1.5 & 1.501 & 1.503 & 1.506 & 1.51 & 1.515 & 1.519 & 1.525 & 1.53\\
		1.5 & 1.632 & 1.679 & 1.811 & 1.892 & 1.969 & 1.949 & 2.082 & 2.021 & 2.009 \\
		1.5 & 1.498 & 1.497 & 1.495 & 1.493 & 1.491 & 1.49 & 1.488 & 1.486 & 1.484
	\end{array}\right.\\
	&\left.\begin{array}{ccccc}
		1.534 & 1.539 & 1.545 & 1.551 & 1.557 \\
		2.034 & 2.127 & 2.105 & 2.215 & 2.35 \\
		1.483 & 1.481 & 1.479 & 1.478 & 1.476
	\end{array}\right]
	\\
	\mathds{X}_+=  &\left[\begin{array}{cccccccccc}
		1.5 & 1.501 & 1.503 & 1.506 & 1.51 & 1.515 & 1.519 & 1.525 & 1.53 & 1.534 \\
		1.632 & 1.679 & 1.811 & 1.892 & 1.969 & 1.949 & 2.082 & 2.021 & 2.009 & 2.034\\
		1.498 & 1.497 & 1.495 & 1.493 & 1.491 & 1.49 & 1.488 & 1.486 & 1.484 & 1.483
	\end{array}\right.\\
	&\left.\begin{array}{ccccc}
		1.539 & 1.545 & 1.551 & 1.557 & 1.565 \\
		2.127 & 2.105 & 2.215 & 2.35 & 2.32\\
		1.481 & 1.479 & 1.478 & 1.476 & 1.474
	\end{array}\right]\\
\end{align*}
\begin{align*}
	\mathds{M}_-=   &\left[\begin{array}{cccccccccc}
		1.5 & 1.5 & 1.501 & 1.503 & 1.506 & 1.51 & 1.515 & 1.519 & 1.525 & 1.53 \\
		1.5 & 1.632 & 1.679 & 1.811 & 1.892 & 1.969 & 1.949 & 2.082 & 2.021 & 2.009  \\
		1.5 & 1.498 & 1.497 & 1.495 & 1.493 & 1.491 & 1.49 & 1.488 & 1.486 & 1.484 \\
		2.25 & 2.447 & 2.521 & 2.723 & 2.85 & 2.973 & 2.953 & 3.163 & 3.081 & 3.073 \\
		2.25 & 2.445 & 2.513 & 2.708 & 2.825 & 2.936 & 2.904 & 3.098 & 3.003 & 2.982  \\
		2.25 & 2.247 & 2.247 & 2.247 & 2.249 & 2.252 & 2.256 & 2.26 & 2.266 & 2.271
	\end{array}\right.\\
	&\left.\begin{array}{ccccc}
		1.534 & 1.539 & 1.545 & 1.551 & 1.557\\
		2.034 & 2.127 & 2.105 & 2.215 & 2.35 \\
		1.483 & 1.481 & 1.479 & 1.478 & 1.476\\
		3.12 & 3.275 & 3.252 & 3.435 & 3.659\\
		3.015 & 3.151 & 3.113 & 3.273 & 3.468\\
		2.275 & 2.28 &  2.286 & 2.292 & 2.299
	\end{array}\right]
\end{align*}
\quad\\\hspace*{2em}{\textbf{\emph{Transformation matrix $	\Theta$, and matrix $P$ designed by} \textsf{SOSTOOLS}.}}\footnote{Since we collected $15$ data points from the unknown Lorenz system, we omit reporting $\mathds{H}(x)$ here for brevity.}
	\small\begin{align*}
		\Theta(x) =\begin{bmatrix}
			1 & 0 & 0 & x_2 & 0 & 0 \\
			0 & 1 & 0 & 0 & x_3 & 0 \\
			0 & 0 & 1 & 0 & 0 & x_1
		\end{bmatrix}^\top\!\!\!\!\!\!,\quad \quad
		\quad P= \left[\begin{array}{ccc} 636.2 & -343.3 & 208.5\\ -343.3 & 1214.0 & -3.128\\ 208.5 & -3.128 & 86450.0 \end{array}\right] \\
	\end{align*}

\end{document}